\documentclass[12pt,a4paper]{article}
\usepackage{cite}
\usepackage{graphicx}
\usepackage[a4paper,left=2.8cm,right=2.8cm,top=2.5cm,bottom=2.5cm]{geometry} 
\usepackage[T1]{fontenc} 
\usepackage{array}

\usepackage{amssymb}
\usepackage{amsmath}
\usepackage{amsthm}

\usepackage{algorithm}
\usepackage{algpseudocode}

\theoremstyle{plain}
\newtheorem{theorem}{Theorem}
\newtheorem{lemma}{Lemma}
\newtheorem{corollary}{Corollary}
\newtheorem{proposition}{Proposition}
\theoremstyle{definition}
\newtheorem{definition}{Definition}
\newtheorem{example}{Example}
\newtheorem{remark}{Remark}

\begin{document}
	
\date{}

\title{Pareto-optimal Nash equilibrium in capacity allocation game for self-managed networks}

\author{Dariusz G\k{a}sior, Maciej Drwal\\ 
	\small Institute of Computer Science\\ 
	\small Wroc\l{}aw University of Technology, Wroc\l{}aw, Poland\\ 
	\small \{dariusz.gasior, maciej.drwal\}@pwr.wroc.pl}

\maketitle


\begin{abstract}
	
	In this paper we introduce a capacity allocation game which models the problem of maximizing network utility from the perspective of distributed noncooperative agents. Motivated by the idea of self-managed networks, in the developed framework decision-making entities are associated with individual transmission links, deciding on the way they split capacity among concurrent flows. An efficient decentralized algorithm is given for computing strongly Pareto-optimal strategies, constituting a pure Nash equilibrium. Subsequently, we discuss the properties of the introduced game related to the Price of Anarchy and Price of Stability. The paper is concluded with an experimental study.

\medskip
\noindent
\textit{\textbf{Keywords:}	computer networks, algorithmic game theory, capacity allocation}
\end{abstract}


\section{Introduction}

Modern communication networks provide universal systems of data exchange within diversified services and applications. Publicly available communication channels, maintained by Internet service providers (ISPs), are shared by very large numbers of concurrent packet flows. Each of such data transmission usually corresponds to the communication between a client application (invoked by a user) and a server application. On the global scale, users can be seen acting independently and willing to selfishly maximize their utility, reflected in their transmission speed or response delay. One distinctive characteristic of such systems is the lack of central coordination or regulation. 

In the presence of limited communication resources, packet transfer protocols need to incorporate congestion avoidance functionalities. It has been shown \cite{kelly2001mathematical} what kind of users' utility can be maximized with the use of Internet transmission control protocol (TCP) \cite{jacobson1988congestion}. Many challenging questions arise when one seeks to design a communication network in such a way so as to optimize a given utility measure. In the last years, this line of research has stimulated many advancements in the area of distributed mathematical optimization \cite{boyd2010distributed}. Some of the most interesting results were obtained with the use of algorithmic game theory, which has become a method of choice for analyzing properties of distributed protocols. 

In this paper we employ this approach to the analysis of distributed transmission rate control problem, formulated within the network utility maximization framework. In the majority of prior research on the network resource allocation games, usually users (clients or flow sources) were modeled as players. Motivated by the idea of self-managed (autonomic) networking, we propose an alternative formulation, in which players are associated only with transmission links. More specifically, in our model, each decision-making agent manages one outgoing router interface, connected directly to some other node. Consequently, with each node there can be associated multiple players, but their decisions correspond to disjoint subsets of flows.

It was argued that the part of network maintenance cost corresponding to human administering/operating of the system is rapidly growing, and becomes negligible in comparison to the devices' prices. It is predicted that such trend will last in the coming years \cite{agoulmine2011autonomic}. Therefore, it is crucial to develop mechanisms which enable managing the network resources in an automatic or semi-automatic manner. The aim of the proposed solution is to limit the human administrator role only to defining the goals of system's operation. It is assumed that network routers may be considered as autonomous entities, which operate independently. Only limited coordinating communication is allowed between them.

\subsection{Related Work}\label{sec:related}

The idea of autonomic networking was introduced by IBM \cite{kephart2003vision}. A similar concept underlies the self-organizing networks (SON) \cite{prehofer2005self}. Recently, SON approach has been extensively studied for the application in 4G LTE mobile networks \cite{hu2010self}. In \cite{mortier2006autonomic} authors argue that some existing protocols like TCP or Open Shortest Path First (OSPF) may be treated as basic solutions for autonomic networking.

The central network resource allocation problem is the network utility maximization problem (NUM), which is also discussed in this paper. As formulated in \cite{kelly1998rate}, it provides the basis for further considerations. An extensive survey of the utility-based approach applied to the analysis of network resource allocation can be found in \cite{chiang2007layering}. Moreover, in \cite{kelly2000models} it was shown that the NUM framework suits well for developing a self-managing mechanism for the Internet.

A survey of the most important approaches in autonomic network management may be found in \cite{dobson2006survey} and \cite{agoulmine2011autonomic}.
However, it is worth noting that the most common concepts towards self-management are related to control theoretic approach \cite{diao2005self}, biological inspired mechanisms \cite{balasubramaniam2006biologically} and game theory \cite{mackenzie2001game}.

Game theory is a very powerful framework for studying decision making problems, involving a group of agents acting individually, being rational and competing or cooperating to achieve certain goals \cite{myerson1997game}. It provides mathematical tools for analyzing the consequences of agents' behavior and enables developing mechanisms, which encourage them to take expected actions. Game theory has been widely studied in the context of many different applications, mostly in economics, but also in politics, biology, philosophy and computer science \cite{leyton2008essentials}. An introductory material on game theory may be found in \cite{easley2010networks, osborne2004introduction}. 

In the recent years a subfield known as algorithmic game theory has emerged \cite{koutsoupias1999worst, papadimitriou2001algorithms}, combining game theory and algorithms design. This was mainly motivated by the need for analysis of interaction of independent agents in the Internet, in such problems as inter-domain routing, peering, online auctions, online advertising, etc. The problems tackled with the use of algorithmic game theory include establishing the existence of Nash equilibria, computing the Price of Anarchy and Price of Stability, and designing computationally efficient procedures for determining players' strategies. Moreover, employing the mechanism design techniques allows for constructing and analyzing computational procedures executed by collections of machines \cite{nisan1999algorithmic}. We refer to \cite{nisan2007algorithmic} as a comprehensive textbook on algorithmic game theory. 

Many interesting results in algorithmic game theory applied to computer networks have been obtained in the last ten years. In \cite{leyton2008essentials} two TCP clients are interpreted as players in the {\it prisoner dilemma} game. Another game-theoretical analysis of TCP is given in \cite{trinh2004game}, focusing on the Vegas version of this protocol. In \cite{kameda2008inefficient} it is shown that the noncooperative games for flow control problems have Pareto-inefficient Nash equilibria. 

In \cite{tang2008game} the co-existence of different congestion avoidance protocols is considered. It is shown that some properties related to the NUM approach do not hold in the presence of heterogeneous congestion signals. Such a situation is explained through game theoretical framework.

The important class of games concerning allocation problems in networks (not necessarily communication networks) are congestion games \cite{rosenthal1973class}. Typically congestion games are applied to routing problems in computer networks, where the  sources (users) are interpreted as players deciding on the selection of paths to transmit data at a given rate \cite{fleischer2004tolls, karakostas2004edge}. The player's strategy consists in deciding how to split this rate among all possible paths from the source to the destination, or, if flows are unsplittable, which routes to use for transmission. In \cite{korilis1995architecting} and \cite{korilis1997capacity} the authors propose a methodology of architecting noncooperative games for network resource allocation problem, which may improve overall system performance during provisioning and operating phase of network lifecycle. The solution is obtained for a parallel link network structure. It is shown that for such a case, the occurrence of the Braess paradox \cite{bean1997braess} may be avoided. Some of these results are extended for a general network. In \cite{johari2003network} the congestion game for the rate allocation problem is presented. The variant of a one-link network is analyzed, and it is shown that for such case, the Price of Anarchy is no greater than 4/3. Similarly, the extension for general networks is briefly discussed.

Bottleneck games are a similar class of routing games, in which a different payoff function is used \cite{harks2010computing}. Although the Nash equilibria for such games usually exist, their performance (estimated via Price of Anarchy values) is usually poor. A game with a relatively low Price of Anarchy is proposed in \cite{kannan2010bottleneck}. In \cite{banner2007bottleneck} two types of bottleneck games are considered, for splittable and unsplittable flows. It is also shown that for both proposed games the Price of Anarchy is unbounded. However, it is proven that under some mild conditions the Nash equilibrium is socially optimal.

Work \cite{larroca2009routing} considers both congestion game and bottleneck game, in the application to the routing problem. It also proposes a new routing game specifically for the elastic flows. All three approaches are compared. Basing on one example and two real network experiments, some advantages of the introduced game are shown.

In \cite{park2000quality} the approach to resource allocation for the networks with quality of service (QoS) based on Differentiated Services \cite{blake1998architecture} architecture is proposed. The sources (flows) are players. They choose one QoS class and the transmission rate in this chosen class. The players' payoffs are proportional to the transmission rate if their QoS requirements are satisfied and zero otherwise. For the proposed noncooperative game, a simple algorithm computing Nash equilibrium is presented. The extension of this concept is given in \cite{fuzesi2002game}.

The joint problem of QoS routing and capacity allocation problem is considered in \cite{elias2010joint}. In the proposed game, two groups of players are introduced, namely capacity players (each related to one link) and network users (each related to one pair of source and destination). Each capacity player divides its capacity among given Class of Services to minimize overall congestion over the associated link. On the other hand each user splits their traffic among all available paths so as to maximize a degree of satisfaction.

In \cite{zhou2010game} the bandwidth allocation problem in the virtual networks (VN) \cite{drwal2011utility} environment is considered. The problem is presented in terms of the non-cooperative game between service providers (VNs' owners) seen as players. The strategy of a player is determined by virtual links' capacities and flow rates in the particular VN. The utility and cost functions constitute the payoffs. The constraints concerning limited amount of physical links' capacities (bandwidths) are substituted with a congestion cost which is one of the addends of the cost function. Authors prove the existence of Nash equilibrium for such a game. An iterative algorithm is proposed, converging to the equilibrium, based on the best response method.

Furthermore, another type of games called auctions \cite{mcafee1987auctions} seems very suitable for computer network applications \cite{koutsopoulos2010auction}.
In \cite{shakkottai2007network} classic Vickerey-Clarke-Groves (VCG) mechanism \cite{makowski1987vickrey}, together with the so-called Kelly mechanism (based on results obtained in \cite{kelly1998rate}), is used for the network resource allocation.
In \cite{dimakis2006mechanisms} the capacity allocation problem is stated as an auction game between flows (users), seen as buyers, and network operator, seen as an auctioneer. A distributed algorithm to find efficient Nash equilibrium is proposed. The presented mechanism is described as VCG-like, since, on the contrary to the classic VCG auction, it does not require a full valuation function. 

Currently, the game-theoretical framework is also extensively studied in the context of wireless networks \cite{felegyhazi2006game}. For instance, in \cite{sahasrabudhe2008bandwidth} bandwidth allocation problem for a class of wireless networks is investigated. The uniqueness of Nash equilibrium for some particular network topologies is shown. It is also stated that some of the presented results may be also generalized for different cases. In \cite{mittal2008game} the problem of choosing an access point by a mobile user is considered from the perspective of this approach. Similar frameworks for issues in wireless networks categorized under corresponding OSI Layers (namely: physical, data link, network and transport layers) are presented in \cite{charilas2010survey}. The VCG auctions were also applied to the wireless networks, e.g. in \cite{fu2007noncollaborative} it is proposed for the resource allocation problem in multimedia wireless networks.

In \cite{saad2009coalitional} coalitional games for a communication systems are considered. 
A classification of such games distinguishing three main types of cooperative games is given.
It is stressed that the need of autonomic and self-organizing networks implies the necessity of developing distributed algorithms which enable each network device to make independent decision concerning network management.
Application examples of cooperative games in computer networks, mainly wireless, are discussed. Presented arguments corroborate the game-theoretic approach as a promising solution for autonomic (and self-organizing) networks.

More detailed surveys of game theoretical applications in various network resource allocation problems may be also found in \cite{altman2006survey, charilas2010survey}.

The solution approach presented in this paper, can be seen as an efficient decentralized heuristic for the network utility maximization problem. There exists a large body of work on this subject; most of these works however focus on exact algorithms, formulated as gradient-based procedures \cite{beck2013optimal}, Lagrangian methods \cite{palomar2006tutorial}, dual decompositions \cite{boyd2010distributed}, Newton-type procedures \cite{athuraliya2000optimization} and interior-point methods \cite{dolev2009distributed}. Since the NUM problem is basically a convex optimization problem, it can be solved in polynomial time. However, due to large scale of practical instances, exact solutions often require considerable computational effort. Moreover, since many models involve uncertain parameters, it is justified to consider approximate solutions. For general multicommodity flow problems approximation algorithms were considered in \cite{leighton1999multicommodity}. For a special type of flow control problems, where formulation can be stated with the use of positive linear programs, a distributed approximation algorithm was given in \cite{bartal2004fast}. Our solution is more general, as a wider class of utility functions are allowed, however only for a restricted subset of instances the approximation bounds are proven; for general network topologies we consider our algorithm a heuristic method, and evaluate its effectiveness via computational experiments.

\subsection{Original contribution}

The main contribution of this paper is an efficient distributed algorithm which provably halts at Pareto-optimal Nash equilibrium of link capacity allocation in the considered utility maximization problem. Unlike typical distributed methods of solving such network problems (e.g. Lagrangian relaxation-based methods for finding saddle point) our approach leads to a fast constructional procedure. Although finding optimal allocation is not guaranteed, computational experiments show that the proposed algorithm is very scalable. It requires only a few iterations regardless of the number of links and flows, and gives a very good approximation of the optimal solution. 

Moreover, this paper introduces a new type of network game, motivated by the idea of self-management. The algorithm is designed to be implemented in a decentralized manner: local information is used by decision-makers, with minimal communication between them achieved via simple mechanisms. The conceptual framework of self-management assumes that each decision-making agent relies on its local information and acts in order to achieve its local goals, which form a decomposition of the global goal, designated by the system designer.

Unlike prior works, e.g.: \cite{altman2002nash, banner2007bottleneck, shakkottai2007network, karakostas2004edge}, in the proposed game decision-making agents control routers' outgoing interfaces (i.e. each player corresponds to one link, and not to end user or flow, as it is usually assumed in the aforementioned literature). This approach makes our solution prone to the negative effects of users' selfishness, limiting their possibilities of abuse. Instead, the proposed algorithm is designed to be implemented by network operators on their router devices in a fully decentralized manner.

While different routers can be under control of independent organizations (e.g.: network operators, ISP companies), they are not obliged to use the same algorithm for allocating transmission rates. Since each decision-making agent is interested in maximizing only the utility of its own services, the selfishness is apparent. Consequently, when designing decentralized algorithms it is crucial to analyze the system's equilibria. Typically, network routing games (such as congestion games or bottleneck games) utilize the notion of Wardrop equilibrium \cite{correa2011wardrop} as a desirable state of the system. In these games, players may choose between alternative paths to transfer their traffic, in order to minimize delays. Moreover, in routing games it is often assumed that there are infinitely many players. However, the game introduced in this paper models a different type of conflict: flow sources do not choose between alternative paths, as these are predetermined. In contrast, individual path components (i.e. links) choose rate allocations for finitely many flows, and the game outcome is defined in terms of network utility. Furthermore, if one would try to transform our game into a kind of routing game, where instead of individual packets, players are considered as aggregations of packets (e.g. originating from a common source), then one would notice that there is significant difference between the strategy spaces. While in both games strategies can be seen as vectors of nonzero allocations, in our game the sum of elements of this vector has to be no greater than given amount, while in routing games it has to be no less than a given amount. 

For such game a natural notion of stability is a pure Nash equilibrium. In such state no link has an incentive to change the rate allocation, as that would not improve link's utility. The utility, in turn, translates into operators' income for providing transmission services.

Finally, we develop preliminary results concerning the quality of Nash equilibria of the game, as compared to the optimal solution of NUM. This allows us to bound approximation ratio of the presented algorithm for a special network topology (i.e. serial network).

\subsection{Organization of paper}

The paper is organized as follows. Section \ref{sec:related} gives an overview of the related work, including references to the game theoretical literature and studies on relevant network resource allocation problems. Section \ref{sec:cag} consists of two parts. The statement of network utility maximization problem is given in Subsection \ref{sec:problem}. The definition of the capacity allocation game is presented in Subsection \ref{sec:def}. Main results of the paper are contained in Section \ref{sec:alg} including two algorithms for computing strategies, and the proof that the strategy profile computed by the latter algorithm (denoted Algorithm \ref{alg2}) is a pure Nash equilibrium. Section \ref{sec:props} discusses additional properties of the game: Pareto-optimality (Subsection \ref{sec:pareto}), Price of Anarchy and Stability (Subsection \ref{sec:poa}). The general results are based on the analysis of hypothetical strategy profile corresponding to the optimal solution of NUM. Computational study is presented in Section \ref{sec:comp}. Finally, Section \ref{sec:concl} concludes the paper.

\section{Capacity allocation game}\label{sec:cag}

We show how the interaction between concurrent decision-making agents can be modeled as a game. Subsequently, we establish a relationship of the formulated game and the solution of the NUM problem.

\subsection{Network utility maximization problem}\label{sec:problem}

In the considered problem the network consists of a set of $L$ links, each with capacity $c_{l} > 0$, $l \in \{1, \ldots, L \}$. Denote ${\bf c} = [c_1, \ldots, c_L]^T$. There are $R$ flows (users' packet transmissions), defined by a routing matrix ${\bf A} = [a_{lr}]$, where $a_{lr}=1$ if $r$-th flow traverses $l$th link, and $a_{lr}=0$ otherwise. Each flow is characterized by the transmission rate $x_r \geq 0$ (expressed in bits per second). Denote ${\bf x} = [x_1, \ldots, x_R]^T$. For each flow there is an associated utility measure $u_r(x_r)$, which is assumed to be strictly increasing concave and twice-differentiable function of transmission rate (reflecting user's willingness to pay their network operator). 

The network utility maximization problem (NUM) introduced in \cite{kelly1998rate} is defined as follows:
\begin{equation}\label{num1}
	\textrm{maximize} \;\;\; Q({\bf x}) = \sum_{r=1}^R u_r(x_r)
\end{equation}
subject to:
\begin{equation}\label{num1:2}
	{\bf A} {\bf x} \leq {\bf c},
\end{equation}
\begin{equation}\label{num1:3}
	{\bf x} \geq 0.
\end{equation}

In this paper we restrict the choice of utility functions to the so-called {\it isoelastic} functions, that is, to the class of functions\footnote{Symbol $\log x$ denotes the natural logarithm function.}:
\begin{equation}\label{isoel}
	u_r(x_r) = \left\{ \begin{array}{ll}
 				w_r \frac{1}{1 - \gamma} x_r^{1 - \gamma} & \gamma > 0, \gamma \neq 1, \\
				w_r \log x_r	& \gamma = 1.
			\end{array}
			\right.
\end{equation}
It was shown that such class of functions leads to proportionally fair allocations of transmission rates \cite{mo2000fair}, thus is typically employed in the analysis of network resource allocation problems.

In practical instances of interest (especially on the scale of Internet autonomous systems or ISP networks), the number of concurrent flows is very large and there is no central authority capable of managing all transmission rates simultaneously. Therefore we are interested in designing and analyzing decentralized protocols, which solve this problem (or approximates its solution) in a distributed manner. Such protocols are typically implemented as a part of low-level operating system's kernel software; in particular TCP/IP stack includes procedures for flow control and congestion avoidance, available in many implementation-dependent variants. Emerging networking solutions provide more advanced means of rate control, incorporating Quality of Service (QoS) capabilities \cite{gasior2008qos, yun2010qos}.

\subsection{Definition of the game}\label{sec:def}


We introduce the following network game. Each $l$th link in the network is associated with one player. Players must decide on the way they allocate their total capacities $c_l$ among the set of flows traversing their corresponding link. Each player makes a decision individually. The decision of $l$th player, called player's strategy, is denoted ${\bf s}_l = [s_{l1}, s_{l2}, \ldots, s_{lR}]^T$, where $s_{lr}$ is the fraction of $l$th link's capacity allocated for $r$th flow. We restrict the player's choice only to feasible decisions, that is, satisfying $\sum_{r=1}^R a_{lr} s_{lr} \leq c_l$. However, the transmission rate of a single flow is limited by the minimal allocation of some link along the path of that flow (defined by the routing matrix ${\bf A}$). The player's payoff is computed as the value of weighted utility of transmission rates of all the flows passing through the corresponding link. The strategy profile of the game is defined as ${\bf S} = [{\bf s}_1,  \ldots, {\bf s}_L]$. The payoff of $l$th player is given by:
\begin{equation}\label{payoff}
	Q_l({\bf S}) = \sum_{r=1}^R a_{lr} b_r u_r \left( \min_{k : a_{kr}=1} s_{kr} \right)
\end{equation}
where $b_r \geq 0$ is a weight assigned to the $r$th flow.

We define {\it social welfare} as the sum of utilities of all flows:
$$
\mathcal{W}({\bf S}) = \sum_{r=1}^R u_r(\min_{l:a_{lr}=1} s_{lr}).
$$
Observe that this is equal to the value of objective function \eqref{num1} of NUM, if transmission rates $x_r$ are computed as minimum of allocations along paths. We make use of this fact in the analysis of Price of Anarchy and Stability in Subsection \ref{sec:poa}. Moreover, the social welfare is also equivalent to the sum of all payoffs \eqref{payoff} with weights:
\begin{equation}\label{weights1}
	b_r = \frac{1}{\sum_{k=1}^L a_{kr}}.
\end{equation}
Here each $r$-th weight is a reciprocal of length of the path associated with $r$th flow. Such form of weights penalize long flows (proportionally to the number of involved links). 

In this paper two types of game's payoffs are considered: {\it uniform payoffs}, that is, all weights $b_r=1$, and the one with the set of weights defined by \eqref{weights1}, called {\it path length payoffs}.

\section{Algorithms for computing strategies}\label{sec:alg}

In the following subsections we present two algorithms for computing certain feasible strategies. In both cases the computations can be carried out independently by all players, since the only constraints imposed on player strategies ${\bf s}_l$ are local. This makes these algorithms suitable for implementation as decentralized protocols. Next, we show that the first algorithm, although simple to implement, does not guarantee establishing an equilibrium. However, the second one, which can be seen as its extension, always finds a pure strategy strongly Pareto-optimal Nash equilibrium.

\subsection{Local one-step allocation algorithm}


Consider the following algorithm. Each player solves a local concave optimization problem, given as:
\begin{equation}\label{alg1}
	{\bf S}^{I}_l = \arg\max_{{\bf S}_l \in D_l} \sum_{r=1}^R a_{lr} b_r u_r(s_{lr})
\end{equation}
where:
$$
	D_l = \left\{ {\bf s}_l : \sum_{r=1}^R a_{lr} s_{lr} \leq c_l, \; \forall_r \; a_{lr} s_{lr} = s_{lr} \right\}.
$$

For the assumed class of utility functions \eqref{isoel}, given fixed $\gamma$, the solution can be derived analytically, as:
$$
	s_{lr} = a_{lr} c_l \frac{(b_r w_r)^{1/\gamma}}{ \sum_{j=1}^R a_{lj} (b_j w_j)^{1/\gamma}}.
$$

This algorithm is a realization of the simplest rational strategy, which can be computed without any communication between players. Due to this fact, there are no synchronization issues concerning implementation in a networked environment. This approach resembles the ``safe'' algorithm for distributed optimization, given in \cite{papadimitriou1993linear} for solving positive linear programs. Moreover, the computational complexity is low, as constructing such solution boils down to solving concave maximization problem in $n_l = | \{ r \in R$ : $a_{lr} = 1 \} |$ variables (in general, this can be accomplished in polynomial time with the use of interior-point methods \cite{boyd2004convex}; for isoelastic utilities evaluating analytic solutions results in $O(n_l)$ time complexity). 

However, as the following example shows, in general this algorithm does not produce a state of equilibrium; a player may be better off changing its allocation without informing other players. 

\begin{example}
	
Consider two links ($L=2$), the first with capacity $c_1= 10$ and the second with capacity $c_2=100$. There are three flows ($R=3$); first flow passes through both links, while the two other flows use single links, link 1 and link 2, respectively (see Figure \ref{fig:ex1}). All utility functions are assumed to be logarithmic, i.e. $u_1(x) = u_2(x) = u_3(x) = \log x$.

The local algorithm computes the following strategy vectors for both players (links), ${\bf S}^I = ({\bf s}^I_1, {\bf s}^I_2)$: 
$${\bf s}^I_1 = (5, 5, 0),$$

$${\bf s}^I_2 = (50, 0, 50).$$

The payoffs of players (i.e. the local utilities of their corresponding flows) are equal $Q_1({\bf s}^I_1) = \log 5 + \log 5$ and $Q_2({\bf s}^I_2) = \log 5 + \log 50$.

If ${\bf S}^I$ were a Nash equilibrium, no player would have an incentive to unilaterally deviate from this allocation. However, since the capacities in both links are uneven, a fair allocation is suboptimal. The high-speed link with $c_2=100$ should promote the flow 3, as it does not pass through the bottleneck link with $c_1=10$. Thus the following change of player 2 strategy:
$$
	{\bf s}'_2 = (5, 0, 95)
$$
gives a better outcome, $Q_2({\bf s}'_2) = \log 5 + \log 95$.

\end{example}

\begin{figure}[!ht]
\begin{center}
    \includegraphics[width=4in]{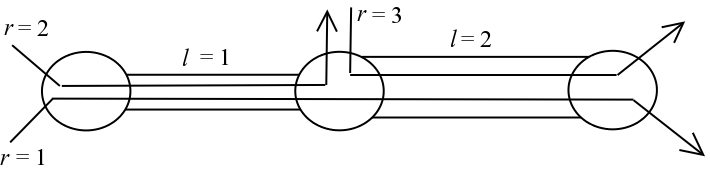}
\end{center}
\caption{Illustration of a simple network setup in Example 1.}\label{fig:ex1}
\end{figure}

\subsection{Iterated allocation algorithm}

The solution computed by the algorithm presented in the previous section can be easily improved, if we allow players to interact in the following way. Initial allocations are computed with the use of the local one-step algorithm. After these allocations are computed, all sources start sending data and transmission rates of all flows gradually increase from zero. The rate of a flow stops accelerating as soon as some link on the path becomes a bottleneck, i.e. the rate reaches minimal allocation of some link along its path. This must be detected with the use of a congestion avoidance algorithm (e.g. as a part of TCP), which notifies the source that the capacity on the path has been exceeded. We call such flow {\it saturated}. This means that it is no longer possible to increase its rate. However, other links on the path of such flow may have unused capacity. Thus it is possible to assign this capacity among the non-saturated flows, increase their rates and repeat that until all the flows become saturated. 

This procedure is summarized as Algorithm \ref{alg2}. Step \ref{alg2:1} requires executing local one-step algorithm (solving problem \eqref{alg1}). In the presented pseudocode, ${\bf S}^{(n)} = [ {\bf s}_l^{(n)}, \ldots, {\bf s}_l^{(n)}]^T$ denotes the strategy vector computed in $n$th iteration. The final strategy vector returned by the algorithm is denoted ${\bf S}^{II}$. The following auxiliary notation is used in the algorithm's description. Denote $\mathcal{R}_l = \{ r : a_{lr} = 1 \}$ the set of flows traversing link $l$. Set $\hat{\mathcal{R}}_n$ denotes all the flows that are not yet saturated in $n$th iteration. Set $\mathcal{L}_n$ denotes all the links that ran out of capacity in $n$th iteration. The smallest index of link that runs out of capacity in $n$th iteration is denoted $\phi(n)$.

The implementation of the presented algorithm needs to be appropriately structured in the networked environment. In particular, step \ref{alg2:3} is executed concurrently by all links $l \in \{ 1, \ldots, L \}$ using local information, while step \ref{alg2:2} can be seen as ``coordinating'' phase, in which saturated flows are detected and removed from further considerations (this is also achieved locally by each link, by detecting whether a flow stopped increasing its rate).

\begin{algorithm}

	\caption{Iterated Allocation Algorithm}\label{alg2}

	\begin{algorithmic}[1]
	
		\State ${\bf S}^{(1)} \leftarrow {\bf S}^I$ \label{alg2:1}
		
		\State $\hat{\mathcal{R}}_1 \leftarrow  \{ 1, \ldots, R \}$
		
		\State $\mathcal{L}_1 \leftarrow \left\{ l : \sum_{r=1}^R a_{lr} \left( \min_{k: a_{kr}=1} s_{kr}^{(1)} \right) = c_l \right\}$
		
		\State $n \leftarrow 2$
		
		\While{$n \leq L$}
		
		\State $\hat{\mathcal{R}}_n \leftarrow \hat{\mathcal{R}}_{n-1} \setminus \mathcal{R}_{\phi(n-1)},$
		
		where:
		$$\phi(i) := \min \mathcal{L}_i,$$
		$$\mathcal{L}_i =
									\left\{ l : \sum_{r=1}^R a_{lr} \left( \min_{k: a_{kr}=1} s_{kr}^{(i)} \right) = c_l \right\}
		$$ \label{alg2:2}
		
		\If {$\hat{\mathcal{R}}_n = \emptyset$}
			\State ${\bf S}^{II} \leftarrow {\bf S}^{(n)}$
			\State \Return ${\bf S}^{II}$
		\EndIf
		
		\State determine ${\bf S}^{(n)}$: 
		$${\bf s}_l^{(n)} = \arg\max_{\bf{s}_l \in D_l^n} \sum_{r \in \hat{\mathcal{R}}_n} a_{lr} b_r u_{r} (s_{lr})$$

		where:
		$$
			D_l^n = \left\{ {\bf s}_l : \sum_{r=1}^R a_{lr} s_{lr} \leq c_l, \; \forall_r \; 0 \leq s_{lr} \leq a_{lr} s_{lr},
			\right.
		$$
		$$
			\left. \forall_{r \in \bigcup_{i=1}^{n-1} \hat{\mathcal{R}}_{\phi(i)}} \; s_{lr} = \min_{k: a_{kr} = 1} s_{kr}^{(n-1)} \right\}
		$$ 
		\label{alg2:3}

		\State $n \leftarrow n + 1$
		
		\EndWhile
		
		\State ${\bf S}^{II} \leftarrow {\bf S}^{(n)}$
		
		\State\Return ${\bf S}^{II}$

	\end{algorithmic}
\end{algorithm}

The correctness of Algorithm \ref{alg2} follows from the following fact:

\begin{proposition}
	For isoelastic utility functions the set
	$$\mathcal{L}_i = \left\{ l : \sum_{r=1}^R a_{lr} \left( \min_{k: a_{kr}=1} s_{kr}^{(i)} \right) = c_l \right\}$$ 
	is nonempty for all $i$, such that $\hat{\mathcal{R}}_i \neq \emptyset$.
\end{proposition}

\begin{proof}
	The proof is by contradiction. Assume that at $n$th iteration of the Algorithm \ref{alg2} there is no link with the capacity completely filled. 
	Since the set $\hat{\mathcal{R}}_n \neq \emptyset$, this implies the existence of sequence of links, constructed as follows. In $m_1$th link there is a flow, which is saturated by link $m_2$. In link $m_2$ there must be some spare capacity, that is, there must be a flow, which is saturated by link $m_3$. Continuing this reasoning, we conclude that there is some link $m_k$, which must belong to the set $\{ m_1, m_2, \ldots, m_{k-1} \}$, since there is only a finite number of links (otherwise $m_k$ must be completely filled, as no other link prevents the increase of its allocation). Without the loss of generality, let $m_k = m_1$, and $r_k = r_1$.
	
	Suppose link $m_1$ saturates the flow $r_1$, that is:
	$$
		s_{m_1r_1} = \min_{k: a_{kr_1} = 1} s_{kr_1} = x_{r_1},
	$$
	and suppose flow $r_2$, which also passes through link $m_1$, is saturated by the link $m_2$. Similarly, link $m_2$ saturates the flow $r_2$, and contains the flow $r_3$, saturated in the link $m_3$, etc. The last of the considered flows, $r_{k-1}$, is saturated in the link $m_{k-1}$, which contains the flow $r_1$.

	The class of isoelastic functions \eqref{isoel} has the property that optimal allocations are proportional to the weights $v_r = (b_rw_r)^{1/\gamma}$. Thus $r$th flow on $l$th link gets the share of capacity $c_l$ equal to $v_{r}$.
	
	The allocations of flows $r_1$ and $r_2$ in link $m_1$ can be written as:
	$$
		s_{m_1r_1} = \frac{v_{r_1} c_{m_1}}{\sum_{k: a_{m_1k} = 1} a_{m_1k} v_k},
	$$
	$$
		s_{m_1r_2} = \frac{v_{r_2} c_{m_1}}{\sum_{k: a_{m_1k} = 1} a_{m_1k} v_k},
	$$
	and consequently:
	$$
		\frac{s_{m_1r_2}}{s_{m_1r_1}} = \frac{v_{r_2}}{v_{r_1}}.
	$$
	Since the flow $r_2$ is not saturated in link $m_1$, there exists a constant $0 < \alpha_1 < 1$, such that:
	$$
		x_{r_2} = \alpha_1 \frac{v_{r_2}}{v_{r_1}} x_{r_1}.
	$$
	Similarly, the allocation ratio of flows $r_2$ and $r_3$ in link $m_2$ can be written as:
	$$
		\frac{s_{m_2r_3}}{s_{m_2r_2}} = \frac{v_{r_3}}{v_{r_2}},
	$$
	and there exists a constant $0 < \alpha_2 < 1$, such that:
	$$
		x_{r_3} = \alpha_2 \frac{v_{r_3}}{v_{r_2}} x_{r_2} = \alpha_2 \alpha_1 \frac{v_{r_3}}{v_{r_1}} x_{r_1}.
	$$
	Continuing this reasoning, we reach the link $m_{k-1}$, in which the allocation ratio is:
	$$
		\frac{s_{m_{k-1}r_1}}{s_{m_{k-1}r_{k-1}}} = \frac{v_{r_1}}{v_{r_{k-1}}},
	$$
	and consequently:
	$$
		x_{r_1} = \left( \prod_{i=1}^{k-1} \alpha_i \right) x_{r_1}.
	$$
	Since the product of $\alpha_i$s is less than $1$, we get a contradiction.
	
\end{proof}

\begin{remark}\label{rem1}
	Strategy ${\bf S}^{II}$ has the property that all link allocations along the path of any flow are equal to the minimum allocation for that flow, i.e.:
	$$
		\forall_r \; \forall_{l : a_{lr} = 1} \; s_{lr}^{II} = \min_{k : a_{kr}=1} s_{kr}^{II}.
	$$
\end{remark}

The algorithm executes no more than $L$ iterations of its main loop. In result, a strategy profile ${\bf S}^{II}$ is returned. Moreover, the strategy ${\bf S}^{II}$ computed by Algorithm \ref{alg2} dominates the strategy computed by local one-step algorithm. 

In each iteration the algorithm completely fills the capacity of a subset of links. However, observe that in each iteration exactly one link (with the smallest index in such subset) is removed from the further consideration. Thus without the loss of generality, it is always possible to renumber the indices of links in such a way, that the order of removed links matches the iteration number, i.e.: $\phi(n) = n$ for all $n=1, \ldots, L$. 

Henceforth, we assume that the links have been renumbered this way, prior to the execution of the algorithm.

\begin{figure}[!ht]
\begin{center}
    \includegraphics[width=3.8in]{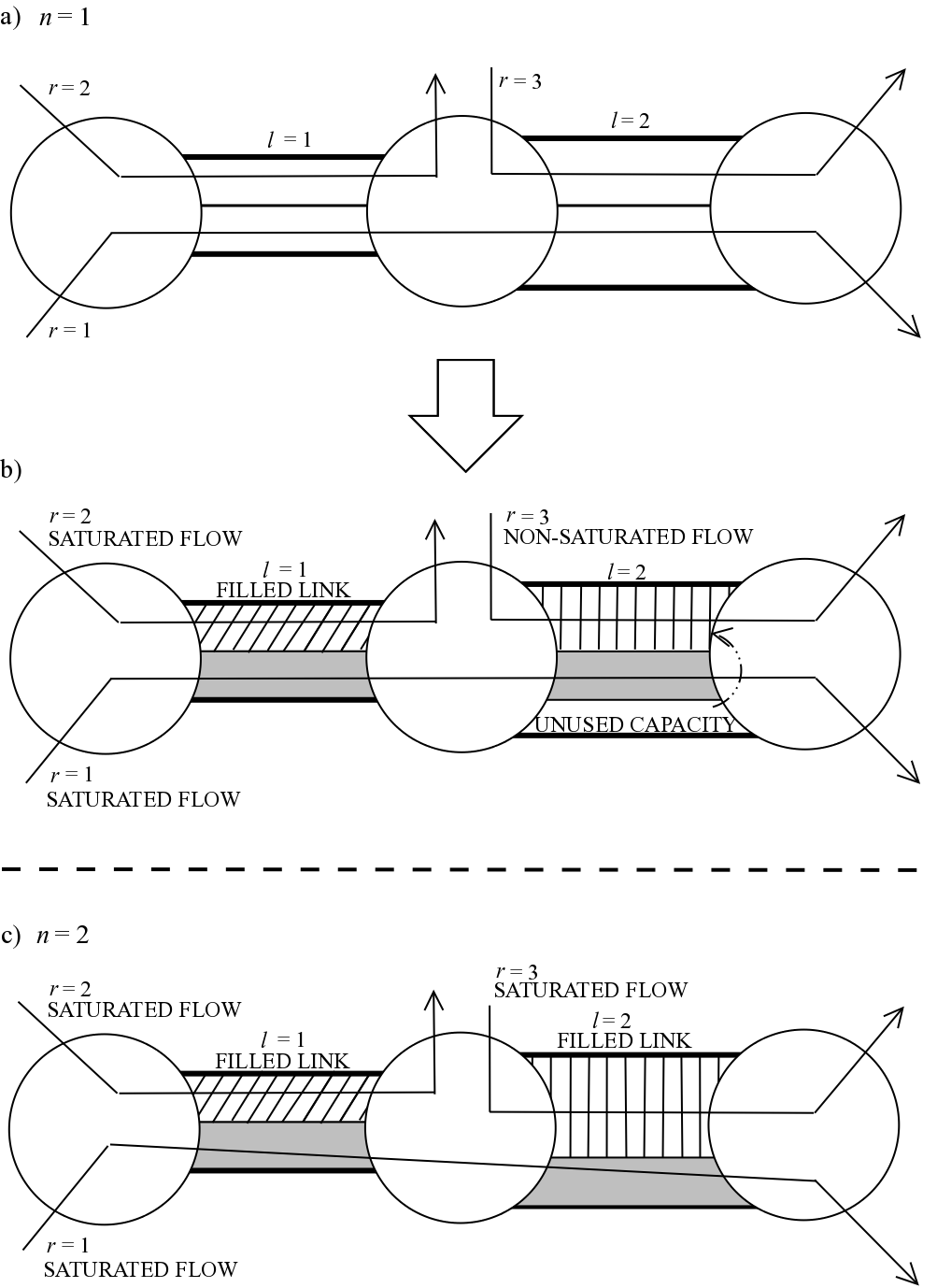}
\end{center}
\caption{The idea of Algorithm \ref{alg2}.}\label{fig:example_dg}
\end{figure}

To illustrate the idea of Algorithm \ref{alg2}, let us consider the following example.

\begin{example}

Assume the same data as in Example 1 (including network structure, number of flows, routing matrix, links' capacities, utility functions) and the game with uniform payoffs.
At the beginning of the execution of Algorithm \ref{alg2} all flows are not saturated, $\hat{\mathcal{R}}_1 = \{1, 2, 3\}$, so each link allocates capacities maximizing its own objective, and the players' strategies are the same as in Example 1, i.e.:
$${\bf s}^{(1)}_1 = (5, 5, 0),$$ 

$${\bf s}^{(1)}_2 = (50, 0, 50),$$
\noindent
as it is depicted in Figure \ref{fig:example_dg}a.
Once all strategies are computed, the sources may transmit data at rates:

$${\bf x} = (5, 5, 50).$$

It is clear that only first link becomes filled ($\mathcal{L}_1 = \{1\}$), and the first and the second flows are saturated (which means that only the third one is not saturated, $\hat{\mathcal{R}}_2 = \{3\}$) as it is shown in Figure \ref{fig:example_dg}b. Consequently, there is unused capacity in the second link, which may be used in the next step.

Now, for $n = 2$ the procedure is illustrated in Figure \ref{fig:example_dg}c. The second link does not change its strategy (allocation) since it is filled, i.e.:
$${\bf s}^{(2)}_1 = (5, 5, 0).$$ 

The second link sets allocations for saturated flows equal to their current transmission rates (i.e. minimal capacity allocated for these flows on their paths) and calculates the allocations for non-saturated flows (in this case for $r = 3$) to maximize its objective. Finally, its strategy is given by:
$${\bf s}^{(2)}_2 = (5, 0, 95).$$

In consequence, only the third source may increase its rate, resulting in rates:
$${\bf x} = (5, 5, 95).$$

Since both stopping conditions are met in the next iteration, i.e. $n = 3 > L = 2$ and $\hat{\mathcal{R}}_2 = \emptyset$, the execution ends with the following players' strategies:
$${\bf s}^{(II)}_1 = (5, 5, 0),$$ 

$${\bf s}^{(II)}_2 = (5, 0, 95).$$

It is easy to check that these strategies constitute a Nash equilibrium.

\end{example}

\subsection{Pure strategy Nash equilibrium}

Observe that utility functions in the considered game are concave (as a composition of isoelastic function and minimum function). From the Rosen's theorem on concave $L$-person games \cite{rosen1965existence} at least one pure strategy Nash equilibrium always exists in this game. It is however not obvious how to compute one efficiently. In this paper we give a polynomial time algorithm which always finds an equilibrium. Our main result is stated as the following theorem:

\begin{theorem}\label{thm:nash}
	A pure strategy Nash equilibrium in capacity allocation game can be computed by Algoritm \ref{alg2} as strategy profile ${\bf S}^{II}$.
\end{theorem}

Before we prove this theorem, we need the following lemmas.

\begin{lemma}\label{lem1}
	Let $f_m : \mathbb{R}_+ \cup \{ 0 \} \rightarrow \mathbb{R}$, $m = 1, \ldots, M$, be continuous, strictly increasing, strictly concave and twice-differentiable functions.
	
	For any $\beta \geq 0$, let us define a vector-valued function $\boldsymbol{\alpha}^* : \mathbb{R}_+ \cup \{ 0 \} \rightarrow \mathbb{R}^M$:
	$$\boldsymbol{\alpha}^*(\beta) = \arg \max_{\boldsymbol{\alpha}} \left\{ \sum_{m=1}^M f_m (\alpha_m) : \sum_{m=1}^M \alpha_m = \beta \right\}, $$
	where $\boldsymbol{\alpha} = [ \alpha_1, \ldots, \alpha_M]^T$.
	
	Each component $\alpha_m^*(\beta)$ is nondecreasing.
	
\end{lemma}

\begin{proof}
	Consider the following optimization problem:
	$$
		\textrm{maximize} \;\; \sum_{m=1}^M f_m(\alpha_m)
	$$
	subject to:
	$$
		\sum_{m=1}^M \alpha_m = \beta.
	$$
	The Lagrange function of the problem is:
	$$
		L(\boldsymbol{\alpha}, \mu) = \sum_{m=1}^M f_m(\alpha_m) + \mu \left( \beta - \sum_{m=1}^M \alpha_m \right).
	$$
	Under the assumptions on functions $f_m$, the problem is concave, thus the KKT conditions \cite{boyd2004convex} imply, that the optimal solution $\boldsymbol{\alpha}^*$, $\mu^*$ satisfies:
	\begin{equation}\label{lem1:1}
		\frac{\partial L(\boldsymbol{\alpha}^*, \mu^*)}{\partial \alpha_m^*} = f'_m(\alpha_m^*) - \mu^* = 0,
	\end{equation}	
	\begin{equation}\label{lem1:2}
		\frac{\partial L(\boldsymbol{\alpha}^*, \mu^*)}{\partial \mu^*} = \beta - \sum_{m=1}^M \alpha_m^* = 0.
	\end{equation}
	
	Since function $f_m$ is twice differentiable and strictly concave, thus the function $f''_m(\alpha_m) < 0$ for all $\alpha_m > 0$. This implies that function $f'_m$ is strictly decreasing and continuous. Hence there exists the inverse function $(f'_m)^{-1}$, which is also strictly decreasing.
	
	From \eqref{lem1:1} the solution must satisfy:
	\begin{equation}\label{lem1:3}
		\alpha^*_m = (f'_m)^{-1}(\mu^*).
	\end{equation}
	Substituting $\alpha^*_m$, for $m=1,\ldots,M$ into \eqref{lem1:2} we obtain:
	\begin{equation}\label{lem1:4}
		\sum_{m=1}^M (f'_m)^{-1}(\mu^*) = \beta.
	\end{equation}

	A sum of strictly decreasing functions is also strictly decreasing, thus the lefthand side of \eqref{lem1:4}, denoted $G(\mu^*) = \sum_{m=1}^M (f'_m)^{-1}(\mu^*)$ is strictly decreasing function of $\mu^*$. Given two values: $\beta_1 = G(\mu^*_1)$ and $\beta_2 = G(\mu^*_2)$, relation $G(\mu^*_1) < G(\mu^*_2)$ implies that $\mu^*_1 > \mu^*_2$. For decreasing $\mu^*$, the value of function $(f'_m)^{-1}$ increases. From relation \eqref{lem1:3}, the corresponding solution $\alpha_m^*$ increases with the increase of $\beta$, for all $m=1, \ldots, M$.
	
\end{proof}

\begin{lemma}\label{lem2}
	Given the strategy ${\bf S}^{II}$ computed by Algorithm \ref{alg2}, consider all strategies ${\bf S}' \neq {\bf S}^{II}$, constructed as follows. For any fixed player $l$, let:
	$$
		s'_{lr} = s_{lr}^{II} + \epsilon_{lr},
	$$
	where each $\epsilon_{lr}$ is any real number. Other players' strategies fulfill: ${\bf S}'_{-l} = {\bf S}^{II}_{-l}$.
	
	For any such strategy ${\bf s}_l'$, if there exists $r \notin \hat{\mathcal{R}}_l$, such that $\epsilon_{lr} > 0$, then there exists a strategy ${\bf s}_l''$, such that:
	$$
		\forall_{r \notin \hat{\mathcal{R}}_l} \; \epsilon_{lr} \leq 0,
	$$
	giving exactly the same payoff.
\end{lemma}

\begin{proof}

We give a constructive proof. Let us denote the set of indices of saturated flows $r \notin \hat{\mathcal{R}}_l$ such that $\epsilon_{lr} \leq 0$ by $\bar{\mathcal{R}}_l^I$ and the set of remaining saturated flows (i.e. for which $\epsilon_{lr} > 0$) as $\bar{\mathcal{R}}_l^{II}$.
Let us consider the strategy ${\bf S}''$ constructed as follows: $s_{lr}'' = s'_{lr}$ for all $r \in \hat{\mathcal{R}}_l \cup \bar{\mathcal{R}}_l^I$ and $s_{lr}'' = s_{lr}^{II}$ for all $r \notin \hat{\mathcal{R}}_l$. Such a strategy does not increase any allocation for the saturated flow.
		
	Since the strategy ${\bf s}'_l$ is feasible, it is easy to see that the new strategy ${\bf s}_l''$ satisfies the constraint ${\bf A}_l {\bf s}_l'' \leq c_l$.
	
	The value of player $l$'s payoff for strategy ${\bf s}_l'$ is:
	$$
		Q_l({\bf s}_l', {\bf S}_{-l}^{II}) =
	$$
	$$
		=\sum_{r \in \hat{\mathcal{R}}_l \cup \bar{\mathcal{R}}_l^I} a_{lr} b_r u_r(\min \{ s'_{lr}, \min_{k \neq l: a_{kr} = 1} s^{II}_{kr} \}) + \sum_{r \in \bar{\mathcal{R}}_l^{II} } a_{lr} b_r u_r(\min \{ s_{lr}', \min_{k \neq l:a_{kr}=1} s_{kr}^{II} \})
	$$
	$$
		=\sum_{r \in \hat{\mathcal{R}}_l \cup \bar{\mathcal{R}}_l^I} a_{lr} b_r u_r(\min \{ s''_{lr}, \min_{k \neq l: a_{kr} = 1} s^{II}_{kr} \}) + \sum_{r \in \bar{\mathcal{R}}_l^{II}} a_{lr} b_r u_r(s_{lr}^{II})
	$$
	$$
		= Q_l({\bf s}_l'', {\bf s}_{-l}^{II}).
	$$

	The second equality follows from the fact that:
	$$
		s_{lr}' = s_{lr}^{II} + \epsilon_{lr} > s_{lr}^{II} = \min_k s_{kl}^{II},
	$$
	which, in turn, follows from the assumption on accordance of numbering iterations and links removed from consideration. Consequently, the payoff from strategy ${\bf s}_l''$ is equal to the payoff from strategy ${\bf s}_l'$.

\end{proof}

Lemma \ref{lem2} immediately implies that all strategies deviating from ${\bf S}^{II}$ and giving better payoffs for player $l$ should not have allocated more capacity for flows saturated before $l$th iteration, than in the strategy ${\bf S}^{II}$. Otherwise, the difference in allocated capacity is wasted; the superfluous fraction of capacity may be used for allocation of the remaining flows in subsequent iterations.

\begin{corollary}\label{cor:1}
	For any $l \in \{ 1, \ldots, L \}$, if there is no strategy ${\bf s}_l'$ such that $\forall_{r \notin \hat{\mathcal{R}}_l} \epsilon_{lr} \leq 0$, for which:
	\begin{equation}\label{col1}
		Q_l({\bf s}_l', {\bf S}_{-l}^{II}) \geq Q_l({\bf S}^{II}),
	\end{equation}
	then there is no other strategy ${\bf s}_l'$ satisfying \eqref{col1}.
\end{corollary}

\begin{lemma}\label{lem3}
	Let ${\bf S}'$ be a strategy constructed as follows. For any fixed player $l$, let:
	$$
		s_{lr}' = s_{lr}^{II} + \epsilon_{lr},
	$$
	where:
	$$
		\forall_{r \notin \hat{\mathcal{R}}_l} \; \epsilon_{lr} \leq 0,
	$$
	and the remaining $\epsilon_{lr}$ are arbitrary real numbers. Other players' strategies fulfill: ${\bf S}'_{-l} = {\bf S}^{II}_{-l}$.
	
	If there exists $r \in \hat{\mathcal{R}}_l$ such that $\epsilon_{rl} < 0$, then there exists a strategy ${\bf s}_l''$, such that $\forall_{r \in \hat{\mathcal{R}}_l} \epsilon_{rl} \geq 0$, giving no lower payoff.
	
\end{lemma}
	
\begin{proof}
	
	The following sequence of inequalities hold:
	\begin{equation}\label{lem3:1}
		\max_{{\bf s}_l \in D'_l} \sum_{r \in \hat{\mathcal{R}}_l} a_{lr} b_r u_r(s_{lr})
		\geq 
		\sum_{r \in \hat{\mathcal{R}}_l} a_{lr} b_r u_r(s_{lr}^{II})
	\end{equation}
	$$
		\geq 
		\sum_{r \in \hat{\mathcal{R}}_l} a_{lr} b_r u_r(\min \{ s_{lr}', s_{lr}^{II} \})
		= Q_{l}({\bf s}_l', {\bf S}_{-l}^{II}) - \sum_{r \notin \hat{\mathcal{R}}_l} a_{lr} b_r u_r(s_{lr}').
	$$
	
	Let:
	\begin{equation}\label{lem3:3}
		\tilde{{\bf s}}_l = \arg\max_{{\bf s}_l \in D_l'} \sum_{r \in \hat{\mathcal{R}}_l} a_{lr} b_r u_r(s_{lr}),
	\end{equation}
	$$
		D_l' = \left\{ {\bf s}_l : \sum_{r \in \hat{\mathcal{R}}_l} a_{lr} s_{lr} \leq c_l - \sum_{r \notin \hat{\mathcal{R}}_l} a_{lr} s_{lr}^{II}, \; \forall_r \; 0 \leq s_{lr}, \forall_{r \notin \hat{\mathcal{R}}_l} \; s_{lr} = s_{lr}^{II} \right\}.
	$$
	\begin{equation}\label{lem3:2}
		{\bf s}_l'' = \arg\max_{{\bf s}_l \in D_l''} \sum_{r \in \hat{\mathcal{R}}_l} a_{lr} b_r u_r(s_{lr}),
	\end{equation}
	$$
		D''_l = \left\{ {\bf s}_l : \sum_{r \in \hat{\mathcal{R}}_l} a_{lr} s_{lr} \leq c_l - \sum_{r \notin \hat{\mathcal{R}}_l} a_{lr} s_{lr}', \; \forall_{r \notin \hat{\mathcal{R}}_l} s_{lr} = s_{lr}', \; s_{lr} \geq 0 \right\},
	$$

	It is clear that $\tilde{\bf s}_l = {\bf s}_l^{II}$ (compare with step \ref{alg2:3} of Algorithm \ref{alg2}).
	
	From the assumption on $\epsilon_{lr}$ for $r \notin \hat{\mathcal{R}}_l$, it holds that:
	$$
		c_l - \sum_{r \notin \hat{\mathcal{R}}_l} a_{lr} s_{lr}' = c_l - \sum_{r \notin \hat{\mathcal{R}}_l} a_{lr} (s_{lr}^{II} + \epsilon_{lr}) 
	$$
	$$
		= c_l - \sum_{r \notin \hat{\mathcal{R}}_l} a_{lr} s_{lr}^{II} - \sum_{r \notin \hat{\mathcal{R}}_l} a_{lr} \epsilon_{lr}
	$$
	\begin{equation}\label{lem3:4}
		\geq c_l - \sum_{r \notin \hat{\mathcal{R}}_l} a_{lr} s_{lr}^{II}.
	\end{equation}
	
	Consider feasible solutions sets $D_l'$ and $D_l''$ projected on a linear subspace restricted only to the coordinates $r \in \hat{\mathcal{R}}_l$ (values of all other coordinates of vectors ${\bf s}_l$ are fixed, although may be different in sets $D_l'$ and $D_l''$). Let us put $\beta_1 = c_l - \sum_{r \notin \hat{\mathcal{R}}_l} a_{lr} s_{lr}'$ and $\beta_2 = c_l - \sum_{r \notin \hat{\mathcal{R}}_l} a_{lr} s_{lr}^{II}$. Observe that the optimal solutions in both cases must fulfill the constraints with equality, as the objective functions are strictly increasing. Since both objective functions \eqref{lem3:2} and \eqref{lem3:3} on the restricted coordinates are identical, and from \eqref{lem3:4} it holds that $\beta_1 \geq \beta_2$, thus from Lemma \ref{lem1}, the corresponding optimal solutions must satisfy:
	$$
		\forall_{r \in \hat{\mathcal{R}}_l} \; s_{lr}'' \geq s_{lr}^{II}
	$$
	which implies:
	$$
		\forall_{r \in \hat{\mathcal{R}}_l} \; \epsilon_{lr} = s_{lr}'' - s_{lr}^{II} \geq 0.
	$$
	
\end{proof}

Combining Corollary \ref{cor:1} and Lemma \ref{lem3} the following can be concluded:

\begin{corollary}\label{cor:2}
 	For any $l \in \{ 1, \ldots, L \}$, if there is no strategy ${\bf s}_l'$ such that $\forall_{r \notin \hat{\mathcal{R}}_l} \epsilon_{lr} \leq 0$ and $\forall_{r \in \hat{\mathcal{R}}_l} \epsilon_{lr} \geq 0$, for which:
	\begin{equation}\label{col2}
		Q_l({\bf s}_l', {\bf S}_{-l}^{II}) \geq Q_l({\bf S}^{II}),
	\end{equation}
	then there is no other strategy ${\bf s}_l'$ satisfying \eqref{col2}.
\end{corollary}

\begin{lemma}\label{lem4}
	Consider any link $l \in \{ 1, \ldots, L \}$ and two flows, one that is saturated in $l$th iteration (denoted $r=1$), and one that is not saturated in $l$th iteration (denoted $r=2$). Denote:
	$$
		[s_{1l}^*, s_{2l}^*]^T = \arg\max_{[x_1,x_2]^T \in D} \left( u_1(x_1) + u_2(x_2) \right),
	$$
	where $D = \{ [x_1,x_2]^T : \; x_1 + x_2 \leq c \}$. Consider any allocation satisfying $s_{1l} \leq s_{1l}^*$ and $s_{2l} \geq s_{2l}^*$. Let $\delta > 0$. Then:
	$$
		0 \leq \left( u_1(s_{1l}^*) + u_2(s_{2l}^*) \right) - \left( u_1(s_{1l}^* - \delta) + u_2(s_{2l}^* + \delta) \right)
	$$
	\begin{equation}\label{lem4:0}
		\leq 
		\left( u_1(s_{1l}) + u_2(s_{2l}) \right) - \left( u_1(s_{1l} - \delta) + u_2(s_{2l} + \delta) \right).
	\end{equation}
\end{lemma}

\begin{proof}
	The first inequality in the claim \eqref{lem4:0} is valid, since $[s_{1l}^*, s_{2l}^*]^T$ is the maximal solution. The second inequality follows from the concavity of functions $u_1$ and $u_2$:
	\begin{equation}\label{lem4:1}
		u_1(s_{1l}^*) - u_1(s_{1l}^* - \delta) \leq u_1(s_{1l}) - u_1(s_{1l} - \delta),
	\end{equation}
	\begin{equation}\label{lem4:2}
		u_2(s_{2l} + \delta) - u_2(s_{2l}) \leq u_2(s_{2l}^* + \delta) - u_2(s_{2l}^*).
	\end{equation}
	Summing \eqref{lem4:1} and \eqref{lem4:2} side by side, and reordering the terms we obtain the claimed inequality \eqref{lem4:0}.
\end{proof}

This lemma states that transferring a fraction $\delta > 0$ of capacity from a saturated flow to a non-saturated one gives a degradation in the value of utility of the selected pair of flows. Moreover, if the allocations deviate from the optimal one (for this pair of flows) in a way that the saturated flow has less capacity, then this degradation is even higher than the one corresponding to the optimal allocation.

The main result of this Section can be now proven.

\begin{proof}[Proof of Theorem \ref{thm:nash}]
		
	Let us consider the following transformation of strategy ${\bf S}$ into $\hat{\bf S}$. Chose any $\delta > 0$, select any link $l \in \{ 1, \ldots, L \}$, select any flow $r \notin \hat{\mathcal{R}}_l$, subtract: $\hat{s}_{rl} \leftarrow s_{rl} - \delta$, select any flow $q \in \hat{\mathcal{R}}_l$, and add $\hat{s}_{ql} \leftarrow s_{ql} + \delta$.
	
	From Corollary \ref{cor:2} it is enough to restrict the considerations only to such strategies ${\bf S}'$, where $s_{lr}' = s_{lr}^{II} + \epsilon_{lr}$, with $\forall_{r \notin \hat{\mathcal{R}}_l} \epsilon_{lr} \leq 0$ and $\forall_{r \in \hat{\mathcal{R}}_l} \epsilon_{lr} \geq 0$.
	
	It is easy to see that any such strategy ${\bf S}'$ can be produced from the strategy ${\bf S}^{II}$ using a finite number of described transformations: only flows $r \in \hat{\mathcal{R}}_l$ may get higher allocations (adding $\delta$), and only $r \notin \hat{\mathcal{R}}_l$ may get lower allocations (subtracting $\delta$). The amount of added and subtracted capacity must be preserved. Let us denote a strategy obtained after applying $k$ such transformations by ${\bf S}^{II(k)}$. 
	
	Therefore, it is enough to show that such transformations satisfy the assumptions of Lemma \ref{lem4}. Consider a flow that gets saturated in $n$th iteration. Its allocations on link $l$ throughout the subsequent iterations of Algorithm \ref{alg2} form the following sequence:
	$$
		s_{rl}^{(1)} \leq s_{rl}^{(2)} \leq \ldots \leq s_{rl}^{(n)} = s_{rl}^{(n+1)} = \ldots = s_{rl}^{II}.
	$$
	This follows from the fact that as subsets of flows passing through link $l$ get saturated, in the subsequent iterations the amount of capacity to distribute among the remaining non-saturated flows form a nondecreasing sequence. From the Lemma \ref{lem1}, the allocations of a particular non-saturated $r$th flow are nondecreasing.
	
	Consider a pair of flows, $r_1 \notin \hat{\mathcal{R}}_l$ and $r_2 \in \hat{\mathcal{R}}_l$, after applying $k$ transformations. Let $n$ be the iteration in which the flow $r_1$ gets saturated. Thus its allocation on link $l > n$ satisfy:
	\begin{equation}\label{thm:nash:1}
		s_{r_1l}^{(n)} \geq \min_{k : a_{kl}=1} s_{r_1k}^{(n)} = s_{r_1l}^{(l)} = s_{r_1l}^{II} \geq s_{r_1l}^{II(k)}.
	\end{equation}
	On the other hand, the allocation on the same link $l > n$ of flow $r_2$ satisfy:
	\begin{equation}\label{thm:nash:2}
		s_{r_2l}^{(n)} \leq s_{r_2l}^{(l)} = s_{r_2l}^{II} \leq s_{r_2l}^{II(k)}.
	\end{equation}
	However, the allocations $[s_{r_1l}^{(n)}, s_{r_2l}^{(n)}]^T$ are the optimal values obtained in the step \ref{alg2:3} of Algorithm \ref{alg2}. Observe that these values are also the optimal solution of the problem of maximizing $u_{r_1}(x_1) + u_{r_2}(x_2)$, with respect to $x_1,x_2 \geq 0$, subject to: $x_1 + x_2 \leq c_l - \sum_{r \notin \{ r_1, r_2 \} } a_{rl} s_{rl}^{(n)}$.
	
	Notice that this optimization problem is equivalent to the one for a pair of flows from the Lemma \ref{lem4}, while the allocation $[s_{r_1l}^{II(k)}, s_{r_2l}^{II(k)}]^T$ deviates from the optimal in the same way as the one in Lemma \ref{lem4} (see relations \eqref{thm:nash:1} and \eqref{thm:nash:2}). Thus any transformation between strategies ${\bf S}^{II}$ and ${\bf S}'$ cannot improve the value of $u_{r_1}(s_{r_{1l}}') + u_{r_2}(s_{r_{2l}}')$. From the additivity of objective function \eqref{num1}, the total utility $\sum_{r=1}^R a_{lr} b_r u_r(s_{rl}')$ cannot increase, and consequently, the payoff $\sum_{r=1}^R a_{lr} b_r u_r(\min_{k: a_{kl}=1} s_{kl}')$ cannot increase.
	
\end{proof}

\section{Properties of the game}\label{sec:props}

In this section we discuss properties of the strategy profile ${\bf S}^{II}$ computed by Algorithm \ref{alg2}, as well as properties of the game itself, in terms of the quality of its pure equilibria. 

\subsection{Pareto optimality}\label{sec:pareto}

Taking a closer look at the construction of strategy ${\bf S}^{II}$ we may conclude that it provides an allocation that cannot improve the payoff of any player without degradation of other payoffs. In particular, we show that the obtained strategy is strongly Pareto-optimal (Pareto-efficient).

\begin{definition}[Pareto optimality]
	Strategy ${\bf S}$ is Pareto-optimal if there is no strategy ${\bf S}' \neq {\bf S}$ such that $\forall_l \; Q_l({\bf S}') > Q_l({\bf S})$.
\end{definition}

\begin{definition}[strong Pareto optimality]
	Strategy ${\bf S}$ is strongly Pareto-optimal if there is no strategy ${\bf S}' \neq {\bf S}$ such that $\forall_l \; Q_l({\bf S}') \geq Q_l({\bf S})$ and there exists such $k$ that $Q_k({\bf S}') > Q_k({\bf S})$.
\end{definition}

Let us observe the following fact:

\begin{proposition}\label{prop2}
	Strategy ${\bf S}^{II}$ is Pareto-optimal.
\end{proposition}

\begin{proof}
	It is enough to consider only link $l=1$. Since the allocation ${\bf s}_1^{II} = {\bf s}^I_1$, and this is the optimal solution of the problem \eqref{alg1}, no strategy can give higher payoff. We conclude that there is no strategy ${\bf S}' \neq {\bf S}^{II}$ such that for all $l$ simultaneously higher payoff can be achieved.
\end{proof}

In order to prove the strong Pareto-optimality, it is enough to show that for no link a strictly higher payoff can be obtained, without decreasing the payoff of any link. 

\begin{theorem}
	Strategy ${\bf S}^{II}$ is strongly Pareto-optimal.
\end{theorem}

\begin{proof}
	The proof is by induction with respect to the increasing sequence of link subsets $A_1 \supset A_2 \supset \ldots \supset A_L$. Initially, consider a single link, $A_1 = \{ 1 \}$, and similarly, as in the proof of Proposition \ref{prop2}, observe that the allocation ${\bf s}_1^{II}$ cannot be changed, as it is the optimal solution of the problem \eqref{alg1}. 

	Let $l \geq 2$. Suppose all flows in the set $A_{l-1}$ cannot have their allocations changed in a way to improve the payoff. Now consider the subset of links $A_l = A_{l-1} \cup \{ l \}$. Flows in link $l$ can be divided into two disjoint subsets: $\hat{\mathcal{R}}_l$ (the flows which allocation is about to be set in $l$th iteration, or later) and $\mathcal{R}_l \setminus \hat{\mathcal{R}}_l$ (the flows which allocation is already fixed in $l$th iteration). As the latter flows are saturated, the increase of their allocations is useless as this cannot contribute to the increase of $l$th link's payoff (due to the transmission rate limit imposed by bottleneck link in $A_{l-1}$). However, if we decrease any of these allocations the payoff of some link in $A_{l-1}$ will have to decrease, which contradicts the inductive assumption. The remaining capacity is allocated in $l$th iteration among the flows in $\hat{\mathcal{R}}_l$ in a way that gives the highest payoff for link $l$ (as in the step \ref{alg2:3} of Algorithm \ref{alg2}). Any deviation from the strategy ${\bf S}^{II}$ of the allocation of flows in $\hat{\mathcal{R}}_l$ would give equal or lower payoffs. Consequently, the allocation ${\bf S}^{II}$ has the property that any change to this strategy cannot improve any links' payoff without causing a loss of at least one of the remaining players' payoffs.
\end{proof}

\subsection{Price of Anarchy and Stability}\label{sec:poa}

In this section we develop some results concerning the quality of social welfare resulting from pure Nash equilibria of the considered game, as compared to the optimal social welfare (which is nearly impossible to achieve in practice, as that would usually require fully centralized planning with the accurate knowledge of all problem parameters by a single decision maker). In contrast, Nash equilibria can be obtained easily in a decentralized way.

For the sake of further analysis, we assume that utility parameter $\gamma$ in \eqref{isoel} is strictly between $0$ and $1$. This is due to the fact that the notions of Price of Anarchy and Stability -- measures of quality of game equilibria \cite{koutsoupias1999worst} -- require the payoffs to be of the same sign (either positive or negative). Consequently, we make use of the following definition:
\begin{definition}
	Let $\mathcal{E}$ be the set of all Nash equilibria of a game. Let ${\bf x}^*$ be the optimal solution of NUM problem \eqref{num1}--\eqref{num1:3}. Let $\mathcal{W}$ be the social welfare function. The Price of Anarchy is defined as: 
	$$PoA = \frac{\max_{\bf S} \mathcal{W}({\bf S})}{\min_{{\bf S} \in \mathcal{E}} \mathcal{W}({\bf S})}.$$
	The Price of Stability is defined as:
	$$PoS = \frac{\max_{\bf S} \mathcal{W}({\bf S})}{\max_{{\bf S} \in \mathcal{E}} \mathcal{W}({\bf S})}.$$
\end{definition}

To begin with, let us consider a reversed problem of game design: given an optimal solution of the NUM problem, does it constitute a Nash equilibrium of some game variant?

\begin{theorem}
	For the uniform payoffs ($\forall_r \; b_r = 1$) the global optimum of problem \eqref{num1}--\eqref{num1:3} is a Nash equilibrium.
\end{theorem}

\begin{proof}
	The proof is by contradiction. All players play the strategy ${\bf S}$, such that $\forall_r \forall_{l: a_{lr} = 1} \; s_{lr} = x_r^*$, where $x_r^*$ is the global solution of the problem \eqref{num1}--\eqref{num1:3}. Consider a player $l \in \{ 1, \ldots, L \}$. Obviously, increasing allocations $s_{lr}$ of flows sharing multiple links does not improve the payoff of player $l$. Consider flows traversing only link $l$, denoted $\bar{\mathcal{R}}_l = \{ r : a_{rl}=1, \; \sum_{k=1}^L a_{rk} = 1 \}$. We call such flows {\it local}.
	
	Links with $\bar{R}_l \neq \emptyset$ must be completely filled (as otherwise increasing any local flow gives higher payoff). Thus in order to increase the allocation of such flow, at least one non-local flow's allocation must be decreased. But on the other hand, if non-local flows yield some capacity for local flows, then from Lemma \ref{lem1} the allocations of all local flows must increase.

	Let us denote $\mathcal{R}_l' = \mathcal{R}_l \setminus \bar{\mathcal{R}}_l$. Let $\epsilon_r > 0$ and $\eta_r > 0$ be taken so as to satisfy:
	\begin{equation}\label{thm:glob:0}
		\sum_{r \in \mathcal{R}_l'} \epsilon_r = \sum_{r \in \bar{\mathcal{R}}_l} \eta_r.
	\end{equation}

	Suppose such change in $l$th player's strategy would improve its payoff:
	\begin{equation}\label{thm:glob:1}
		\sum_{r \in \mathcal{R}_l'} u_r(x_r^* - \epsilon_r) + \sum_{r \in \bar{\mathcal{R}}_l} u_r(x_r^* + \eta_r) > \sum_{r \in \mathcal{R}_l} u_r(x_r^*).
	\end{equation}
	The value of global objective can be expressed by adding to the both sides of \eqref{thm:glob:1} the utilities of all flows that do not pass through $l$th link:
	$$
		\sum_{r \in \mathcal{R}_l'} u_r(x_r^* - \epsilon_r) + \sum_{r \in \bar{\mathcal{R}}_l} u_r(x_r^* + \eta_r) + \sum_{r \in \{1, \ldots, R \} \setminus \mathcal{R}_l } u_r(x_r^*)
	$$
	\begin{equation}\label{thm:glob:2}
		> \sum_{r \in \{ 1, \ldots, R \} } u_r(x_r^*) = \max_{{\bf x} : \; {\bf A}{\bf x} \leq {\bf c}} \sum_{r \in \{ 1, \ldots, R \} } u_r(x_r).
	\end{equation}
	The last expression is the global optimum of \eqref{num1}--\eqref{num1:3}. From \eqref{thm:glob:0}:
	$$
		\sum_{r \in \bar{\mathcal{R}}_l} (x_r^* + \eta_r) + \sum_{r \in \mathcal{R}_l'} (x_r^* - \epsilon_r)  = \sum_{r \in \bar{\mathcal{R}}_l \cup \mathcal{R}_l'} x_r^* \leq c_l,
	$$
	which shows that the lefthand side of \eqref{thm:glob:2} is a value of feasible solution. Hence the inequality in \eqref{thm:glob:2} gives a contradiction.
\end{proof}

This immediately implies that the strategy giving optimal solution is also the best possible equilibrium:

\begin{corollary}
	For the uniform payoffs the Price of Stability is $1$.
\end{corollary}

Local flows (i.e. ones occupying only a single links), introduced in the above proof, play a key role in the attainability of equilibria. The following example shows that the optimal solution of NUM occurs at a non-equilibrium point of the game variant in which weights are proportional to flow lengths \eqref{weights1}.

\begin{example}
	
Consider the game with weights $b_r = 1/ \sum_{k=1}^L a_{kr}$, with $10$ players, corresponding to links of equal capacity, $c_1 = \ldots = c_{10} = 6$. As depicted in Figure \ref{fig:ex2}, there are two flows: one that passes through all links, with utility function $u_1(x_1) = 10 \log x_1$, and one local, with utility function $u_2(x_2) = 2 \log x_2$. The global optimum is $x_1^* = 5$ and $x_2^* = 1$.

The payoff of $1$st player for allocations $s_{11}=5$ and $s_{12} = 1$ is:
$$\frac{u_1(x_1^*)}{\sum_{k=1}^{10} a_{k1}} + \frac{u_2(x_2^*)}{\sum_{k=1}^{10} a_{k2}}  = \frac{10}{10} \log 5 + \frac{2}{1} \log 1 = \log 5.$$
However, changing strategy of player $1$ to $s_{11}=2$ and $s_{12} = 4$ gives a higher payoff: $\log 2 + 2 \log 4 = \log 2 + \log 16 > \log 5$.

\end{example}

\begin{figure}[!ht]
\begin{center}
    \includegraphics[width=5.3in]{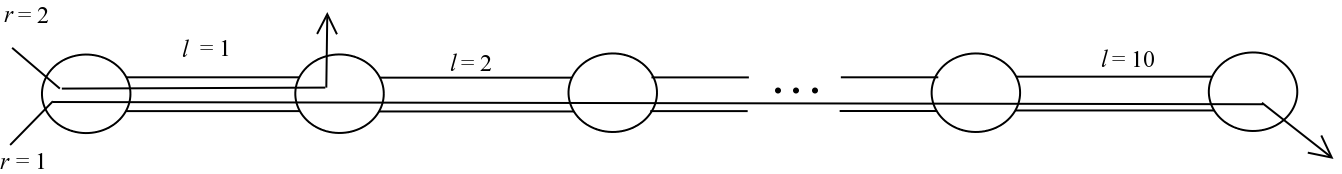}
\end{center}
\caption{Illustration of Example 3: one flow passes through all links, one flow is local.}\label{fig:ex2}
\end{figure}

\begin{corollary}
    The global optimum of problem \eqref{num1}--\eqref{num1:3} is not necessarily a Nash equilibrium of a game with path length payoffs (i.e. weights $\forall_{r} \; b_r = 1 / \sum_{k=1}^L a_{kr}$).
\end{corollary}

One may observe that, provided there are no local (single-link) flows, it is easy to characterize strategies giving pure Nash equilibria, even for a more general class of games (including the one considered in this paper, see Remark \ref{rem1}).

\begin{theorem}\label{thm:local}
	If there are no local flows, i.e. for all $r \in \{ 1, \ldots, R \}$, $\sum_{l=1}^L a_{lr} > 1$, any strategy profile such that:
	$$
		\forall_{l_1, l_2} \; \forall_{r: a_{l_1r} = a_{l_2r} = 1} \; s_{l_1r} = s_{l_2r}
	$$
	constitutes a pure Nash equilibrium of the considered game (regardless the values of $b_r$).
\end{theorem}

\begin{proof}
	Fix $l \in \{ 1, \ldots, L \}$. Any change of $l$th player's strategy vector ${\bf s}_l$ into ${\bf s}_l'$, such that $s_{lr}' = s_{lr} + \epsilon$, $\epsilon > 0$, cannot increase the payoff, unless all the players $l'$ sharing the path with $l$ also increase allocation: $s_{l'r}' = s_{l'r} + \epsilon$. Any decrease in allocation, $s_{lr}' = s_{lr} - \epsilon$, $\epsilon > 0$, can only decrease the local payoff.
\end{proof}

Although the strategies characterized in Theorem \ref{thm:local} are points of equilibrium of the considered game, they can be arbitrarily bad in terms of the players' outcomes. Observe that a zero-allocation strategy ($\forall_{l,r} \; s_{lr} = 0$) is a Nash equilibrium if there are no local flows. Consequently, in such case the Price of Anarchy is unbounded.

While in general the game has inefficient equilibria, it is possible to bound the Price of Anarchy for some important special cases. The following network topology can be seen as fragment of larger structure that occur repeatedly in complex networks.

\begin{example}

Let us consider the following network structure. There are $L \geq 2$ links of equal capacity connected serially, as depicted in Figure \ref{fig:ex4}. There are $L+1$ flows, where each link contains one local flow, and there is a single ``long'' flow passing through all links (this flow has weight $w_{L+1}$). For such network setup, it is possible to analytically derive the value of PoA. It is enough to observe that the worst equilibrium can occur in one of two cases: either when each player allocates zero for the ``long'' flow and $C$ for local flow; or when each player allocates $k = \max_{l=1,\ldots,L} s^{(I)}_{ll}$ for local flow and the remaining $C-k$ for the ``long'' flow. Which one of these is the worst case depends on the weights of flows. Let $W = \sum_{r=1}^L w_r$ and $\omega = \max_{l=1,\ldots,L} w_l$. In the first case we obtain:
$$
	PoA_1 = \left( \frac{w_{L+1}^{\frac{1}{\gamma}}}{W^{\frac{1}{\gamma}}} + 1 \right)^{\gamma},
$$
while in the second case:
$$
	PoA_2 = \left( w_{L+1}^{\frac{1}{\gamma}} + W^{\frac{1}{\gamma}} \right)^{\gamma} \left( b_{L+1} w_{L+1}^{\frac{1}{\gamma}} + \omega^{\frac{1}{\gamma}} \right)^{1-\gamma} \left( b_{L+1} w_{L+1}^{\frac{1}{\gamma}} + W \omega^{\frac{1}{\gamma}-1} \right)^{-1}.
$$
The actual value of PoA is equal to $\max\{ PoA_1, PoA_2 \}$.

\end{example}

\begin{figure}[!ht]
\begin{center}
    \includegraphics[width=5in]{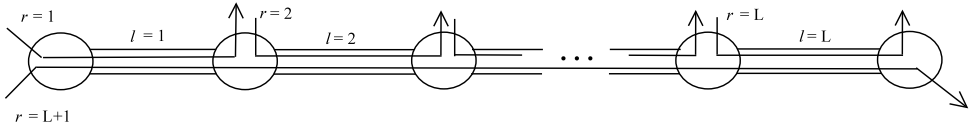}
\end{center}
\caption{Illustration of serial network setup from Example 4.}\label{fig:ex4}
\end{figure}

Let us define $\chi = \frac{w_{L+1}}{W}$, that is the ratio of importance of ``long'' flow to local flows. This enables us to conclude the following:

\begin{corollary}
	The Price of Anarchy for a serial network with $L$ local flows and one flow of length $L$ (Figure \ref{fig:ex4}) satisfies:
	
	1) if $\chi \rightarrow 0$ then PoA $\rightarrow 1$,
	
	2) if $\chi = 1$ then PoA $ \leq (1 + b_{L+1}^{1/\gamma})^{1-\gamma}(b_{L+1}^{1/\gamma})^{-1}$ (in particular, for uniform payoffs, PoA $ \leq 2$),
	
	3) if $\chi \rightarrow \infty$ then PoA $\rightarrow \infty$.
\end{corollary}

The above statement applies to both uniform and path length payoffs. One interesting case is when $\chi = 1$, that is, the sum of local flow weights is equal to the weight of ``long'' flow. From 2) it can be seen that for uniform payoffs ($b_{L+1}=1$) the PoA is no greater than 2. However, in the game variant which penalizes ``long'' flow ($b_{L+1} = 1/L$) the PoA grows to infinity with increasing $L$.

\section{Computational experiments}\label{sec:comp}

In this section we present an experimental study based on a prototype implementation of Algorithm \ref{alg2}. Routing matrices used in the experiments were generated randomly, with probability of a flow passing through a link equal to $0.5$. Link capacities were in the range $10$--$100$ Mbps. Utility functions were of the form \eqref{isoel} with parameter $\gamma = 0.5$. 

\begin{figure}[!ht]
\begin{center}
    \includegraphics[width=4.2in]{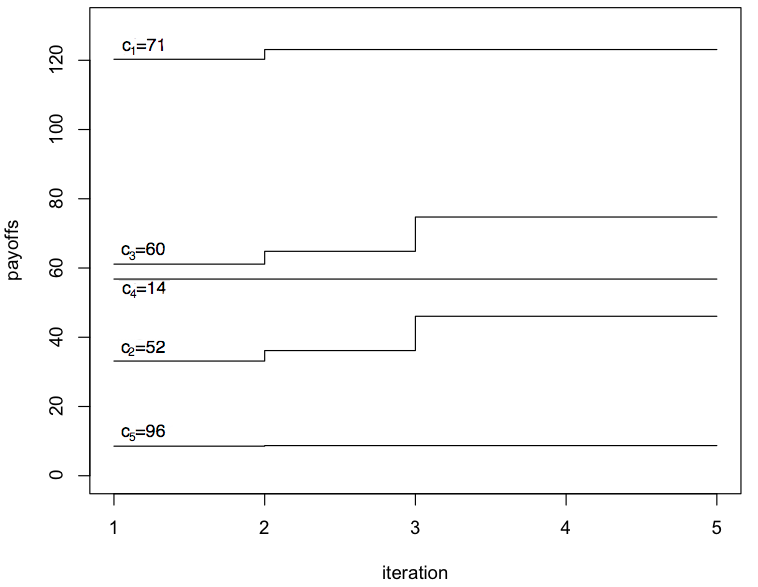}
\end{center}
\caption{Small network: payoffs of players in subsequent iterations.}\label{fig:ce1}
\end{figure}

\begin{figure}[!ht]
\begin{center}
    \includegraphics[width=4.2in]{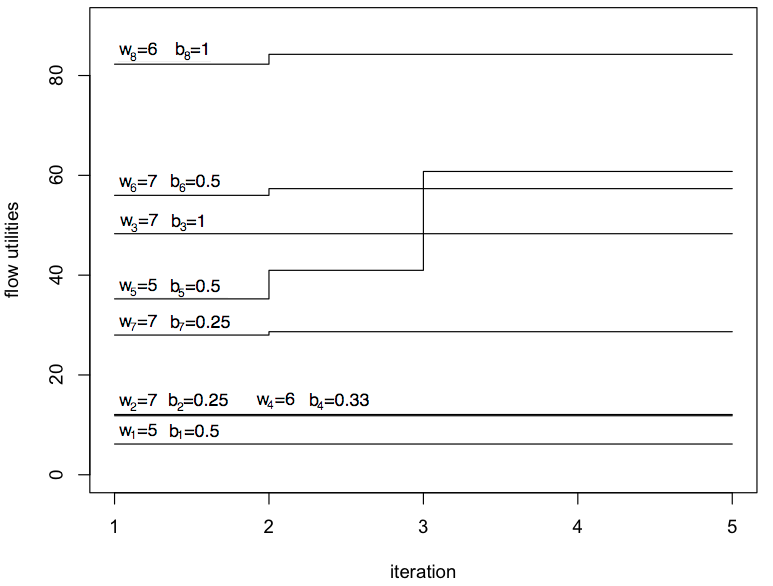}
\end{center}
\caption{Small network: utilities of flows in subsequent iterations.}\label{fig:ce2}
\end{figure}

The first experiment involves a small example of $5$-link network with $8$ flows. For the game variant with weights inversely proportional to path lengths, Figure \ref{fig:ce1} shows how players' payoffs changed in the subsequent iterations, while Figure \ref{fig:ce2} shows the corresponding changes of flows' utilities. It can be seen that only $4$ out of $8$ flows get improved by Algorithm \ref{alg2}, as compared to the initial solution computed by one-step local algorithm. Moreover, algorithm stabilizes all transmission rates after just $3$ iterations. Figure \ref{fig:ce3} compares the total utility \eqref{num1} obtained in this process with its value obtained by the considered algorithm in the game variant with uniform payoffs (all weights $b_r = 1$). The latter variant gives a slightly better solution, but it overestimates the social welfare, which leads to a small degradation of solution in case of larger networks. Additionally, both solutions are compared to the optimal one, computed by directly solving NUM problem \eqref{num1}--\eqref{num1:3} using interior-point method solver.

\begin{figure}[!ht]
\begin{center}
    \includegraphics[width=4.2in]{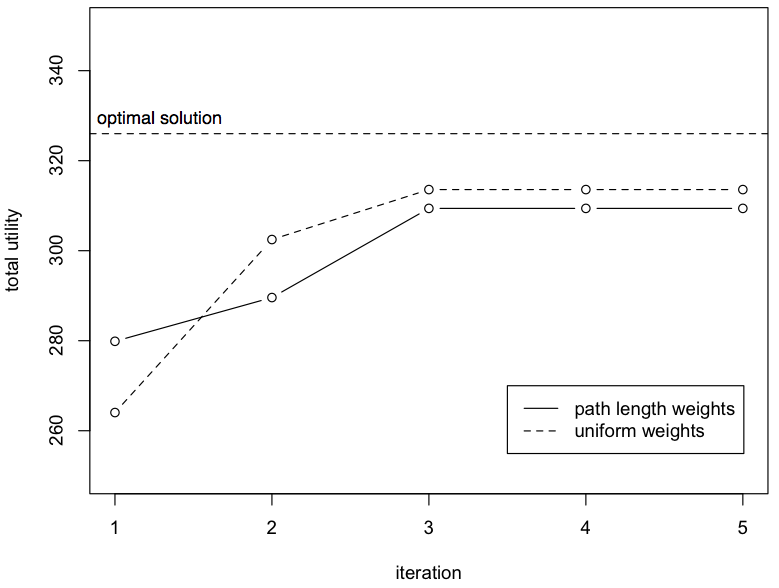}
\end{center}
\caption{Small network: comparison of total utility (NUM objective) obtained from Algorithm \ref{alg2} in two game variants.}\label{fig:ce3}
\end{figure}

\begin{figure}[!ht]
\begin{center}
    \includegraphics[width=4.2in]{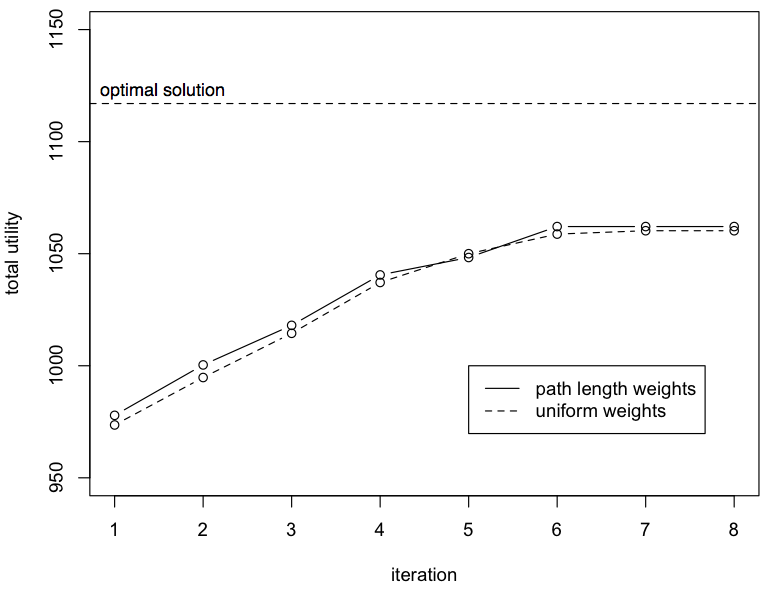}
\end{center}
\caption{Large network: comparison of total utility (NUM objective) obtained from Algorithm \ref{alg2} in two game variants.}\label{fig:ce4}
\end{figure}

In the second experiment there were $100$ links and $200$ flows. Figure \ref{fig:ce4} presents the comparison of solution value changes throughout the iterations, compared to the optimal solution. In the variant of weights inversely proportional to path lengths the algorithm stopped after $6$ iterations, while for uniform weights it took $7$ iterations. Both solutions are similar and close to the optimum (only about $5\%$ worse).

The above results show the huge speedup that can be achieved with the use of the presented algorithm, as compared to the state of the art exact gradient-based algorithms for NUM. In particular, \cite{beck2013optimal} gives an algorithm that has asymptotically faster convergence rate than typical gradient-based methods, but still requires about 1000 iterations to reach $5\%$ region around the optimum even for 5-link network. In contrast, our procedure reaches such accuracy in 3 iterations for 5-link network and in 7 iterations for 100-link network. It should be noted that execution time of a single iteration in both algorithms is very similar.

In \cite{wei2010distributed} authors proposed to use a distributed variant of Newton method for solving NUM on the same network topology as in \cite{beck2013optimal}. The algorithm required about 100 iterations to reach $5\%$ region around the optimum. However, a single step in distributed Newton method is more expensive than one step computation in our method.

Concluding, from the practical point of view our method (which in general case is a heuristic) is better suited for real-time distributed optimization problems, such as the ones encountered in computer networks. It should be noted that the implementation used in the presented experiments neglects some technical details related to real networks. In particular, detecting the flow saturation based on the measurements of rate variability requires a careful treatment and considering real-time programming issues.

\section{Conclusions}\label{sec:concl}

The distributed optimization algorithm presented in this paper provides a fast and scalable method of solving network utility maximization (NUM) problem. By considering links (network routers) as decision-making agents, it extends the idea of solving the problem separately for each of them, by iteratively improving local solutions, based on the detection of minimal rate allocation along flows' paths. Considering this algorithm as a method of computing a strategy for the introduced capacity allocation game, we proved that it finds a strongly Pareto-optimal pure Nash equilibrium. Since in the considered game links are considered players, and payoffs are interpreted as total utilities of all flows passing through a link, the objective of NUM coincides with the sum of payoffs (social welfare), if the utilities of flows are weighted inversely proportional to flow path lengths.

We also proved that in the game variant with equal weights of all flows, the considered game has a pure Nash equilibrium at the point of optimum of NUM. Unfortunately, if there are no flows occupying single links, then it is possible to construct a pure Nash equilibrium resulting in arbitrarily bad solution (i.e. the Price of Anarchy diverges to infinity). Apart from these general results, a detailed analysis concerning Price of Anarchy has been carried out for a special serial network structure.

Computational experiments show that the presented algorithm is very efficient, especially compared to exact gradient methods, and requires a very small number of allocation updates in order to reach an equilibrium, regardless of the network size and number of concurrent flows. Solutions even for very large randomly generated networks were within a few percent of optimal value.

\bibliographystyle{plain}
\bibliography{paper}

\end{document}